\documentclass[
submission
]{dmtcs-episciences}
\usepackage{graphicx,amssymb,makecell}
\usepackage[fleqn]{amsmath}
\newtheorem{theorem}{Theorem}[section]
\newtheorem{observation}[theorem]{Observation}
\newtheorem{definition}[theorem]{Definition}
\newtheorem{lemma}[theorem]{Lemma}
\newtheorem{corollary}[theorem]{Corollary}

\newtheorem{problem}[theorem]{Problem}

\title{Recognising the overlap graphs of subtrees of restricted trees is hard}
\author{Jessica Enright \affiliationmark{1} and Martin Pergel\affiliationmark{2}\thanks{
Partially supported by the Czech Science Foundation grant GA19-08554S.}}
\date{\today}
\affiliation{
University of Edinburgh, Easter Bush, Midlothian, UK.\\
Department of Software and Computer Science Education (KSVI), Charles University, Prague, Czech Republic. 
}
\begin{document}
%
%
\newenvironment{my_enumerate}{
\begin{enumerate}
  \setlength{\itemsep}{1pt}
  \setlength{\parskip}{0pt}
  \setlength{\parsep}{0pt}}{\end{enumerate}
}
\newenvironment{my_itemize}{
\begin{itemize}
  \setlength{\itemsep}{1pt}
  \setlength{\parskip}{0pt}
  \setlength{\parsep}{0pt}}{\end{itemize}
}
\long\def\fixme#1{{\bf FIXME!!!} #1}
\maketitle
\begin{abstract}
The overlap graphs of subtrees in a tree (SOGs) generalise many other graphs classes with set representation characterisations.   The complexity of recognising SOGs in open.  The complexities of recognising many subclasses of SOGs are known.  
We 
consider several subclasses of SOGs by restricting the underlying tree.   For a fixed integer $k \geq 3$, we consider:
\begin{my_itemize}
 \item The overlap graphs of subtrees in a tree where that tree has $k$ leaves
 \item The overlap graphs of subtrees in trees that can be derived from a given input tree by subdivision and have at least 3 leaves
 \item The overlap and intersection graphs of paths in a tree where that tree has maximum degree $k$
\end{my_itemize}

We show that  the recognition problems of these classes are NP-complete. For all other parameters we get circle graphs, well known to be polynomially recognizable.
\end{abstract}
\section{Introduction}
Intersection graphs of geometric objects 
are both theoretically and practically important. 
Their uses include applications
in VLSI-circuit design and ecology.
A graph $G=(V,E)$ with a vertex set $V=\{v_1,\ldots v_n\}$ and the
edge set $E=\{e_1,\ldots, e_m\}$ is an intersection graph of a set
system $\{s_1,\ldots s_n\}$, where for all $i$, $s_i\subseteq \cal S$,
each vertex $v_i$ corresponds
to a set $s_i$ and each edge $e=(v_i,v_j)$ is equivalent to the
fact that $s_i\cap s_j\not= \emptyset$. Intersection graphs -- particularly geometric intersection graphs, where the sets are defined in a geometric way and the set relationship corresponding to adjacency is containment, intersection, or overlapping -- have been explored for many years and research includes characterisations of many classes defined int his way, as well as algorithms exploiting their representations.  As the intersection graphs are very well known, in the introduction we focus mainly on other classes, first mentioning results on containment graphs and then focussing on overlap graphs.  Before doing that, let us yet notice that this text is a full version of the extended abstract \cite{extabstr} and its main goal (among improved introduction) is to show the details of the constructions.

\subsection{Containment graphs}

A possible type of similarly defined classes are containment graphs, where, again, sets represent the vertices and an edge corresponds to the fact that one set is a subset of the other. This way of representation has interesting consequences for poset-theory and it was used, for example, to show that the recognition of posets of dimension 3 and height 2 is hard \cite{FMP}. 

Further possibility of exploring posets is via subtree-containment graphs. Obviously, each poset can be represented as a subtree-containment graph in the following way: We pick its realizer, i.e., minimum set of permutations, whose intersection the poset is and we make a path of each such permutation. Then we pick yet one vertex and make it adjacent to the maxima of all permutations. Each subtree representing an individual vertex is delimited by the vertices (on each permutation) having the appropriate label. An element is above another one, if on all permutations the appropriate label lies above the other one and, therefore, the corresponding subtree is a subtree of the other one, for an illustration, see Figure~\ref{fig:realrepr} showing how to obtain a subtree containment representation from three given linear orders.
\begin{figure}
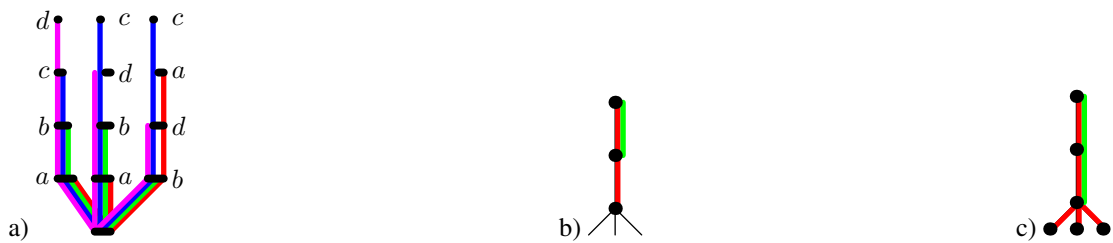

a) \includegraphics{sog.5}\hfill b) \includegraphics{sog.6}\hfill c) \includegraphics{sog.7}
\label{fig:realrepr}
\caption{a) How to get a subtree containment representation from a realizer. b), c) How to extend a subtree not containing a given (branching) node. In order to extend the green one, we extend the red one in the depicted way, too, to avoid making an accidental overlap of the red and green subtree.}
\end{figure}

This construction shows one implication of a theorem that should be available, for example, in \cite{leclerc} saying that posets with dimension at most $k$ are exactly subtree containment graphs in trees with (at most) $k$ leaves. As the referenced article is written in French and the verification of the theorem would require the reader to understand French and to make our article more self-contained, we add yet the sketch of the other implication, as it can be easily observed:

Considering a subtree containment representation, we can derive the realizer. First, let us note that if there is a vertex in the underlying tree that is contained in all subtrees and if no vertex of the underlying tree is a leaf of more than one subtree, we are done (we derive the realizer by traversing from this vertex along paths to individual leaves and the order of leaves for individual subtrees gives us the appropriate permutations). Now, it remains to modify each subtree containment representation into this form. This can be done in two steps. First (that, in fact, we perform as the second one) is that given a representation where two subtrees end, we can change it so that each subtree ends in a different vertex (we add a new vertex on an edge incident with such a vertex and either shrink or extend one of the subtrees -- and this we perform repeatedly while necessary). The last operation we perform to obtain the desired representation is that if not vertex meets our requirements, we pick one of the branching vertices (i.e., a vertex of degree at least 3) and we start extending the trees that do not reach this vertex. The extending can be performed so that we pick a greedily maximal subtree that does not reach this vertex, i.e., no other subtree not yet containing the central vertex ends on the path where this subtree should be extended and then we extend this vertex to contain the central vertex. Again, we may repeat the process as necessary. While reaching the central vertex, we extend the subtrees slightly into all paths stemming from the central vertex in the same way as when disambiguating the endvertices. In such a way we obtain the desired representation. Subtree containment graphs are still being explored \cite{AGG,GL}

\subsection{Overlap graphs}
Similarly, a graph $G=(V,E)$ with a vertex set $V=\{v_1,\ldots v_n\}$ and the
edge set $E=\{e_1,\ldots, e_m\}$ is an overlap graph of a set
system $\{s_1,\ldots s_n\}$, where for all $i$, $s_i\subseteq \cal S$,
each vertex $v_i$ corresponds
to a set $s_i$ and each edge $e=(v_i,v_j)$ is equivalent to the
fact that $s_i\cap s_j\not= \emptyset$ and neither $s_i \subset s_j$ nor $s_j \subset s_i$.

When we consider the overlap and intersection graphs of particular types of set systems, we define graph classes.
Part of the theoretical interest in geometric intersection and overlap  graphs stems from efficient algorithms for otherwise NP-hard problems on these graph classes.
Often, these algorithms require as input a set intersection representation of a particular type.
 Thus we are interested in 
whether or not a given graph has a particular type of intersection representation.
This is called the {\em recognition problem.} Now, let us explain the state of art mainly for intersection graphs, as these are (up to our knowledge) best explored among these three families.

\subsection{Relations between interseciton and overlap graphs}
Probably the oldest intersection-defined graphs are {\em interval
graphs}, the {\em intersection graphs of interval on a line} \cite{GH}.
The interval graphs are generalised by {\em intersection graphs of paths in a tree} \cite{Fel,MonmaW86}.
Intersection graphs of paths in a  tree are in turn generalised by {\em chordal graphs}.
While primarily defined as the graphs without induced cycles of greater than three, chordal graphs are also exactly the {\em intersection graphs of subtrees in a tree} \cite{gavril74}.  
The overlap analogue of chordal graphs is the class of  {\em subtree overlap graphs}, the {\em overlap graphs of subtrees in a tree}.   Subtree overlap graphs generalize many set representation characterized classes, including chordal graphs and therefore interval graphs.

Gavril \cite{gavril2000} defined {\em interval filament graphs} and {\em subtree filament graphs} as intersection graphs of {\em filaments on intervals} and {\em filaments on subtrees}, respectively.
{\em Filaments}
are {\em curves above some geometric structure} (in this case above intervals or subtrees) such that
filaments above disjoint structures must not intersect, while filaments
above overlapping structures ({\em i.e.,} over sets $a$ and $b$ such that
$a\cap b, a\setminus b$ and also $b\setminus a$ are non-empty) must
mutually intersect.

Interval filament graphs are a subclass of subtree overlap graphs, and subtree filament graphs are exactly subtree overlap graphs \cite{Jess}.

Given a set representation, we can solve some otherwise hard problems on these classes, including many problems on chordal graphs \cite{tarjanCan}, and maximum weighted clique and independent set on subtree filament graphs and interval filament graphs \cite{gavril2000}.

Recognising interval filament graphs is known to be hard \cite{gavril2000,Perm}.  In contrast, we can recognise interval graphs and chordal graphs in linear time \cite{corneilOlariuStewart,tarjanCan}, and intersection graphs of paths in a tree in $O(nm)$ time, where $n$ is the number of vertices and $m$ the number of edges in the input graph \cite{AAS}.  
The complexity of recognizing subtree overlap graph is open.

With this in mind, we define three overlap subclasses of subtree overlap graphs:  
we define {\em $k$-SOG} as the {\em overlap graphs of subtrees in a tree such that the tree has at most
$k$-leaves},  class {\em $k$-degree-POG} as the {\em overlap graphs of subpaths in a tree such that the tree has maximum
degree at most $k$}, and the class {\em T-SOG} as the {\em overlap graphs of subtrees of a trees derived from an input tree $T$ by subdivision of edges.}

Though we expect the recognition of subtree overlap graphs to be NP-complete, we expected the recognition of these simplified SOGs to be polynomial time.
We were therefore surprised when 
we obtained hardness results for the recognition problems of 
 $k$-SOG and
$k$-degree-POG for fixed integer $k \geq 3$ and for $T-SOG$ provided that $T$ has at least three leaves.   We present these hardness results in this paper.
The result about $k$-degree-POG also holds for corresponding
intersection graphs; our reduction also shows that it is
NP-complete to 
recognise intersection graphs of paths in a  tree with a fixed maximum degree greater than two.
In contrast,  intersection graphs of subpaths in a tree  can be recognised in polynomial time. \cite{gavril1978}

Before we proceed to the presentation of our results, let us compare our results with other similar classes and let us start with subtree containment graphs: The results presented earlier in the introduction show that subtree containment graphs can be recognised in a polynomial time (\cite{McCSpin} as for a given graph we can decide in a linear time whether it has a transitive orientation and therefore whether it is a poset). Whether a graph can be represented as a subtree containment graph in a tree with at most $k$ leaves is linear for $k=2$ \cite{McCSpin,Spin} and NP-complete otherwise - even for triangle-free graphs \cite{YAN,FMP}. Recognition of containment graphs subpaths in a tree is yet open \cite{GL} even without restriction on maximum degree of the underlying tree.

Our result on the hardness of recognizing the subtree overlap graph with $k$ leafage for fixed integer $k \geq 3$ provides a counterpoint to 
work on the intersection leafage of chordal graphs.  
Stacho and Habib \cite{HS} give a polynomial-time algorithm for determining the leafage of a chordal graph and constructing a representation that achieves that leafage. Leafage was further explored by Chaplick and Stacho \cite{steveStacho} in terms of vertex-leafage where each subtree in the representation is permitted to have at most $k$ leaves by showing that the vertex-leafage is polynomially solvable for vertex-leafage at most 3 and NP-complete otherwise. In contrast with the former and as another pebble into mosaic of leafage, we show that determining the leafage of a subtree overlap graph is NP-Hard for leafage at least 3.

There has been substantial work on the intersection graphs of subtrees or paths in a tree with parameterisation of the subtrees or the underlying tree.

Jamison and Mulder \cite{jamisonMulder} considered the intersection graphs of subtrees of a tree parameterised by maximum degree of both the underlying tree and the individual subtrees.  They showed that the intersection graphs of subtrees of a tree in which the subtrees and the underlying tree have bounded maximum degree $3$ are exactly chordal graphs, and so can be recognised in linear time.  This contrasts to our work, which shows that recognizing the intersection graphs of paths in an underlying tree of bounded maximum degree $3$ in NP-hard.  

Golumbic and Jamison \cite{golumbicJamison} showed that recognising the edge-intersection graphs of paths in a tree is NP-complete, and showed that on a tree with maximum degree three, the edge-intersection and vertex-intersection graphs are the same classes.  

Golumbic et al \cite{golumbicLipsteynStern} explore the complexity of recognizing the intersection graphs of paths in a  tree parameterised by both the maximum degree of the underlying tree and the number of vertices that must be shared between two paths for them to be considered as intersecting.    They provide a complete hierarchy of graph classes using these parameters. 

The idea of a simpler representation, as well as previous work by \cite{pergelKratochvil} motivated us
to define complicacy for subtree overlap graphs:
For a subtree overlap graph $G$, its {\em complicacy} is {\em minimum $k$, such
that $G$ is a $k$-SOG}. We denote this complicacy by $\hbox{cmp}_S(G)$.
For a natural number $n$, by $\hbox{cmp}_S(n)$ we denote the {minimum $k$
such that every subtree-overlap graph on $n$ vertices is also a $k$-SOG}.
Due to Cenek \cite{eowynThesis}, it holds that $\hbox{cmp}_S(n)\le n$. As a minor result, we obtained
a lower bound $\hbox{cmp}_S(n)\ge n-\log n+o(\log n)$.

\begin{table}
\centering
{\bf Summary of results related to the recognition of graphs with geometrical representation by possibly restricted subtrees in a~tree}
\begin{tabular}{|l|l|l|l|}
\hline
~& intersection & overlap & containment\\
\hline
subtrees&P\cite{TarjanYannakakis}&?&P\cite{AGG}\\
\hline
subpaths&P \cite{gavril1978}&?&?\\
\hline
	leafage&P \cite{HS}&{\em NP-C}&NP-C\cite{YAN,FMP,leclerc}\\
\hline
	\makecell{subpaths in a tree with\\ restricted degree}&{\em NP-C}&{\em NP-C}&?\\
\hline

\end{tabular}
\caption{Leafage and degrees must be at least 3 to gain the hardness. The italic shows our results.}
\end{table}


The table suggests that the intersection and containment graphs of subtrees in a tree are better explored than the overlap ones. Usefulness of the overlap graphs can be witnessed, e.g., by the fact that the overlap classes use to generalize the other ones (usually, from the intersection or containment representation, we can easily derive a similar overlap one).


\section{Preliminaries and Definitions}
All graphs discussed here are simple, undirected, loopless and finite.  If  $G = (V, E)$ is a graph, and $V'\subset V$ a vertex subset, then $G[V']$ denotes the subgraph of $G$ induced by $V'$.  We generally adhere to notation in \cite{golumbicBook}.

To simplify our discussion of relationships between sets, we introduce some basic notation.
  Let $s_i$ and $s_j$ be two sets:
if $s_i \cap s_j = \emptyset$ then $s_i | s_j$, pronounced $s_i$ is \textbf{disjoint} from $s_j$, and 
if $s_i \cup s_j \neq \emptyset$, $s_i \backslash s_j \neq \emptyset$ and $ s_j \backslash s_i \neq \emptyset$ then  $s_i \between s_j$, pronounced $s_j$ \textbf{overlaps} $s_j$.

Let $G = (V, E)$ be a graph represented by subsets $\mathcal{S}$ of set $S$.  Then for convenience, we refer to the set corresponding to vertex $v_i \in V$ as $s_i \in \mathcal{S}$.

Let $t$ be a subtree of tree $T = (V_T, E_T)$.  A \textbf{boundary node} of $t$ is a node of $T$ that is in $t$, and has a neighbour in $T$ that is not in $t$.
 \textbf{SUB($T$)} is the set of all trees that can be derived from $T$ by subdividing the edges of $T$ any number of times.
A \textbf{twig} of tree $T$ is a maximal induced path of $T$
that includes a leaf of $T$ and no node of $T$ that has degree greater than two.
A \textbf{lastbranch} of $T$ is a node $p$ of $T$ of degree at least three such that the forest formed by $T[V_T \backslash p]$ has at most one connected component that  is not a path. Figure~\ref{fig:notion} gives an example of lastbranches and twigs.
Unless otherwise noted, if we say that a graph $G$ \textbf{can be represented on a tree $T$} we mean that there exists  an overlap representation of $G$ with $T$ as the underlying tree. 
\begin{figure}
\centerline{\scalebox{1.5}{\includegraphics{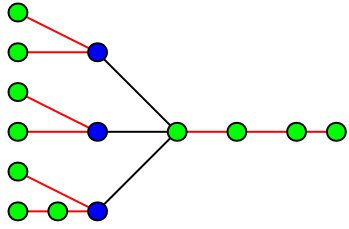}}}
\caption{A coloured example illustrating some of the notions used in our proofs. \textbf{Lastbranches} are the vertices shown in blue - note that the vertex shown in green is not a lastbranch as it is incident to only one twig.  The edges of subgraphs that we call \textbf{twigs} are shown in red. By red edges we denote subgraphs that we reference as twigs. Green branching vertex is not a lastbranch as it is incident only to one twig (instead of three).}
\label{fig:notion}
\end{figure}

Our reductions will be based on blocking off parts of the host tree, thus forcing some vertices to be represented on the twigs of that host tree.  We therefore define a convenient notion of a \textbf{nice} (subtree overlap) representation for an vertex set that requires adjacent members of the vertex set to be represented on different twigs of the host tree.  It is important to note that for this definition of niceness, we will be considering multiple graphs: the overlap graph of our representation, and another graph with a superset of the original's vertices.  The notion of niceness would be meaningless in relation to the overlap graph of our representation, as subtrees corresponding to adjacent vertices must overlap, and so cannot be represented entirely on different twigs of the host tree, this forcing them to be disjoint.  
Let $G = (V, E)$ 
and $G'' = (V'', E'')$ be graphs such that there exists vertex set $V' \subset V$ and $V' \subset V''$ and $G''$ is the overlap graph of  subtrees $\mathcal{T}''$ of tree $T''$.  
Let $\mathcal{T}' \subset \mathcal{T}''$ be the subtrees corresponding to the members of $V'$.
 We say that $V'$ is  \textbf{nicely represented with respect to $G$} if every member of $\mathcal{T}'$ is contained in a twig of $T''$, and there are no two members $v_i, v_j$ of $V'$ such that $(v_i, v_j) \in E$ and $t_i$ and $t_j$ intersect with the same twig of $T''$.  

If a vertex set $V'$ is nicely represented with respect to a graph $G$, then each twig of the nice representation corresponds to an independent subset of $V'$ in $G$, and so:

\begin{observation}\label{obs:niceRepImpliesColouring}
Let $G = (V, E)$ be a graph.  If there exists a representation of some graph on tree $T$ with $k$ twigs in which vertex subset $V' \subset V$ is nicely represented with respect to $G$, then there is a $k$-colouring of $G[V']$.
\end{observation}

When reasoning about nice representations, we will use the idea of an \emph{interpath}: if $I$ and $J$ are two intervals on line $L$, then the {\em interpath} of $I$ and $J$ is a subset of line $L$ that contains $I$, $J$ and the portion of $L$ between them.

The main contribution of this work is to show several hardness results. To do so, we use $k$-colouring of a 3-connected graph as the source problem in our reductions.  Although this problem is very likely already known to be NP-complete, the authors have been unable to find a published proof of this, so we include a sketched proof for completeness.

\begin{problem}
The problem 3-CON-k-COL($G$) is the decision problem for a fixed integer $k \geq 3$: for a 3-connected graph $G$, is there a vertex colouring of $G$ using  $k$ colours? 
\end{problem}
\begin{theorem}
3-CON-k-COL is NP-Complete for fixed integer $k \geq 3$.
\label{perm:theorem-on-3-colouring}
\end{theorem}
\begin{proof}

Firstly, the problem is in NP, as it is a decision problem and a colouring would serve as a certificate.

Now, we supply a reduction from $k-$colouring a connected graph.  Let $G = (V, E)$ be a connected graph.  Let $G_1=(V_1, E_1), G_2=(V_2, E_2)$ and $G_3=(V_3, E_3)$ be three disjoint copies of $G$.  If $v$ is a vertex in $V$, let $v_i$ be its copy in $G_i$, for $i \in \{1, 2, 3\}$. 
We now connect these three copies by adding edges between $v_1, v_2, v_3$ for every vertex $v \in V$: informally, we have joined up all copies of each original vertex.  Let $G' = (V', E')$ be $G_1, G_2, G_3$ plus these connecting edges.  
For any vertex pair $u_i, v_j$ where $u \neq i$ there are three disjoint paths corresponding to a single $uv$ path in $G$: each passing through a different copy of $G$. To break all of them, we have to remove at least 3 vertices.  For any vertex pair $v_i, v_j$ there are there are two paths directly through the triangle formed by $v_1, v_2, v_3$, and at least one additional path through the copies in $G_i, G_j$ of any neighbour of $v$ in $G$.   Therefore,  $G'$ is 3-connected and has minimum degree 3.  An example of this construction and the disjoint paths mentioned here can be found in Figure~\ref{fig:3con}.
\begin{figure}
\centerline{\scalebox{1.5}{\includegraphics{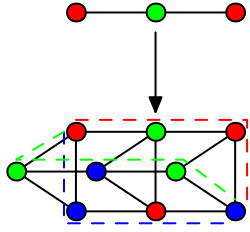}}}
\caption{Idea how to reduce (normal) $k$-colouring to 3-CON-$k$-COL. We reduce $P_2$. We also show, how to find three disjoint paths for one pair of vertices (dashed). Other cases are also very simple. Colours of the vertices correspond to the assigned colours (in problem of colouring).}
\label{fig:3con}
\end{figure}

It remains to 
show that
$G'$ is $k$-colourable for $k \geq 3$ if and only if $G$ is $k-$colourable.

First, assume that $G'$ is $k$-colourable.  Since $G$ is isomorphic to an induced subgraph of $G'$, then $G$ is $k$-colourable.  
Next, assume that $G$ is $k$-colourable with colours $0, 1,\ldots k-1$. We make a colouring of $G'$ in such a way that if a vertex $v$ has a colour $c(v)$, vertex $v_i$ is assigned a colour $c(v_i)=c(v)+i \mod k$, i.e., we circularly shift the colours for each of $V_1, V_2, V_3$. As the number of colours is at least 3, we do not assign the same colour to any pair of copies of the same vertex.
\end{proof}

Having described our source problem, we now define the problems that are the focus of this work, and that we will show are NP-complete:

\begin{problem}
The problem REC-PMD-$k(G)$ is the decision problem for a fixed integer $k \geq 3$:  does there exist a family of paths $\mathcal{T}$ of tree $T$ with maximum degree $k$ such that $G$ is the overlap graph of $\mathcal{T}$?
\end{problem}
\begin{problem}The problem REC-SUB-T($G$) is the decision problem for a fixed tree $T$ with at least three leaves: does there exist a tree $T' \in SUB(T)$ and a family $\mathcal{S}$ of subtrees of $T'$ such that $G$ is the overlap graph of $\mathcal{S}$?
\end{problem}
\begin{problem}
The problem REC-T-$k(G)$ is the decision problem for a fixed integer $k \geq 3$:  does there exist a tree $T$ with $k$ leaves and a family $\mathcal{S}$ of subtrees of $T$ such that $G$ is the overlap graph of $\mathcal{S}$?
\end{problem}

We can conclude that these problems are in NP just by reference:
Cenek \cite{eowynThesis} showed that every minimal subree overlap representation of a graph $G$  is of size polynomial in the size of $G$, and could be checked for correctness in polynomial time.  This would serve as a certificate, so we can conclude that  REC-SUB-T($G$), REC-T-$k(G)$, and REC-PMD-$k(G)$  are in NP.

Our main results are reductions from 3-CON-k-COL($G$) to REC-SUB-T($G$), REC-T-$k(G)$, and REC-PMD-$k(G)$.

All of our reductions work on the same intuition.  We reduce from an instance of graph colouring    
to an instance of a representability problem.

We start with an instance of 3-CON-k-COL($G$), and transform this graph into another                                                                                                          
graph $G''$, which will be our instance of the representability problem.

In the transformation from $G$ to $G''$, the vertices and edges of $G$ are used as vertices of $G''$.  
Other gadgets are added to $G''$ in such a way that
  for REC-SUB-T($G$) and REC-T-$k(G)$,                                                                                               
 twigs in the host tree  of the representation correspond to  colours used for a $k$-colouring                                                                                               
of graph $G$.  That is, no two vertices corresponding to adjacent vertices in $G$ are represented on the same twig of the representation of $G''$.

Similarly for REC-PMD-$k$,  we start with an instance of 3-CON-k-COL($G$), and transform it to another graph $G''$ that will be our instance of REC-PMD-$k$.  $G''$ includes the vertices and edges of $G$ in its vertex set.  
Given a representation of $G''$, we use subtrees of the host tree that would form connected components of the forest that would result from removal of a 
fixed vertex $v$ of degree $k$ 
to correspond to colours used for a $k$-colouring of graph $G$.
 No two gadgets corresponding to adjacent vertices in $G$ are represented on the same branch of the representation of $G''$. That is, their representing subtrees cannot contain $v$.  We do this by forcing all vertices of $G''$ that are edges of $G$ to contain $v$, and by using another class of gadget vertices.  
 
 The overarching idea is that for each of the reductions to the problems we are considering we associate sections of the host tree with one of $k$ colours, and force all vertices of $G''$ that are also vertices of $G$ to be represented in one of these sections.  
To deal with a technical complication, the vertex set of $G''$ includes several copies of the vertex set of $G$ as we show that a constant number                                                                                             
of vertices may be represented in a bad way.  By adding more                                                                                             
copies than the number of possible bad exceptions, at least one copy                                                                                           
must fulfill our requirements, allowing us to derive a colouring of $G$ from a representation of $G''$.

The rest of this manuscript is organised as follows: in Section 2 we define a gadget that is used in all three reductions, and prove a number of utility lemmas on representations of that gadget.  In Sections 3 and 4 we describe each of the three reductions in turn, using a common overall strategy.  Finally, in Section 5 we make some concluding statements and suggestions for future areas of research.

\section{$G_d^u$-blocking gadget and the construction of $G''$}

The purpose of the $G_d^u$ gadget (depicted in Figure~\ref{fig:GDU}) is to occupy parts of the tree in a representation, forcing vertices in the copies of the vertex set of the original graph to be represented on twigs.  We can then map these twigs to colours, generating a colouring for our original graph. In this section we will describe $G_d^u$ and its use in producing  $G''$ for a given graph $G$. This graph $G''$ will be used in all our reductions.  Because we are using colouring of 3-connected graphs as the source problem in our reductions, we may restrict our attention to 3-connected $G$.

Let $d \geq 3$ and $u \neq 1$ be two positive natural numbers. We use these as parameters in building a $G_d^u$ graph: they indicate the number of paths between different named vertices in the graph. $G_d^u$ is consists of four named vertices $v_s, v_b, v_{s'}$ and $v_{b'}$, joined by specified numbers of paths:  vertices $v_s$ and $v_b$ are connected by $d$ disjoint paths of length 3 (with the middle vertex of the first of these paths being $v'_b$), and vertices $v_{s'}$ and $v_{b'}$ are joined by $u$ disjoint paths of length 3. 

\begin{figure}
\scalebox{0.5}{\includegraphics{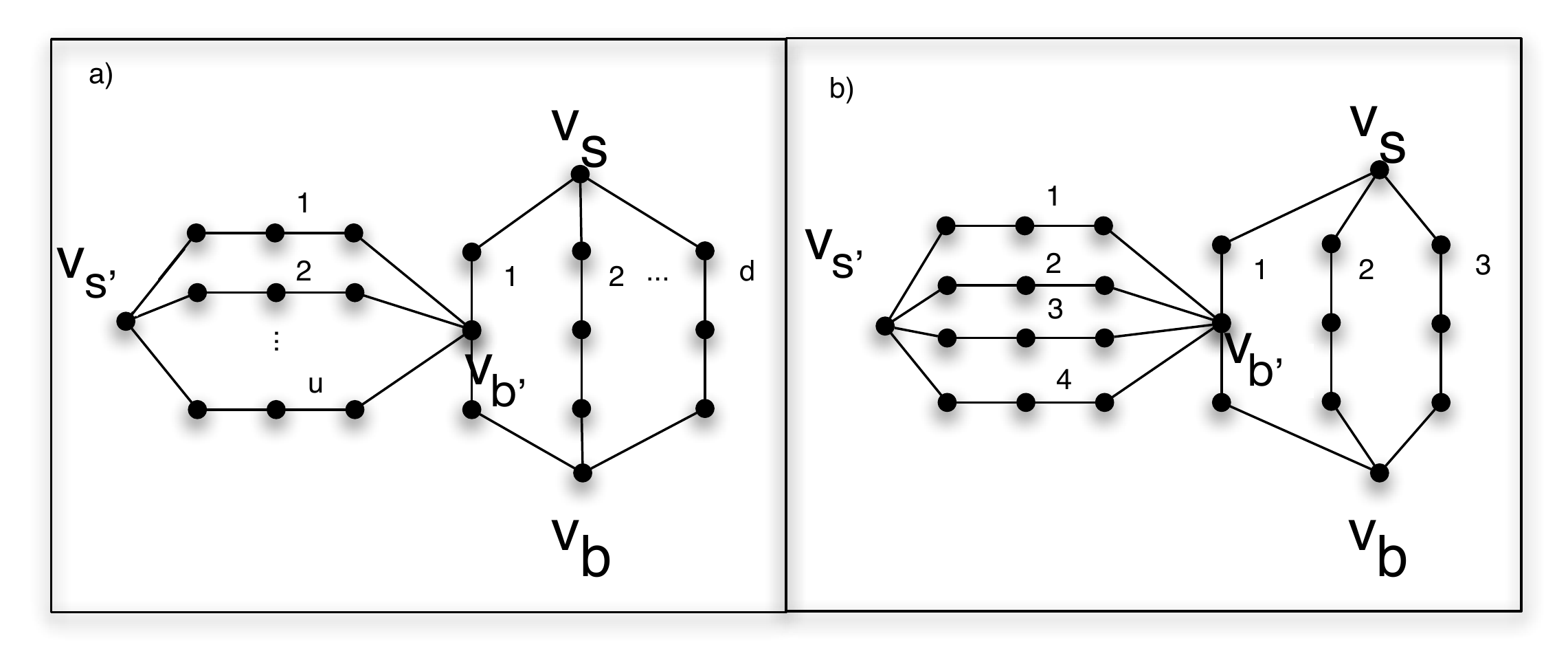}}
\caption{In a) the $G_d^u$ graph - note the presence of $d$ paths of three vertices between vertices $v_s$ and $v_b$, and $u$ paths of three vertices between $v_{b'}$ and $v_{s'}$.  In b) an example: the $G_3^4$ graph.}
\label{fig:GDU}
\end{figure}

Let $G = (V, E)$ be a 3-connected graph. 
 Let $G'$ be the disjoint union of six copies of $G$.  For later convenience, we refer to the vertices of the six isomorphic connected components of $G' = (V', E')$ as $V_a, V_b, V_c, V_d, V_e, V_f$, or collectively as $V_a ... V_f$.

We  describe the production of a graph $G''$ from $G'$.
We define four vertex sets:
\begin{my_enumerate}
\item $V_1 = V'$
\item $V_2 = E'$
\item $V_3 = $ a set of size $|V'|$ disjoint from $V_1 \cup V_2 \cup V_4$
\item $V_4 =$ the vertices of $G^u_d$ (Figure \ref{fig:GDU})
\end{my_enumerate}
Then let $V'' = V_1 \cup V_2 \cup V_3 \cup V_4$, and $f$ be a bijection between $V_1$ and $V_3$.
Each of $V_1, V_2, V_3$ will serve a particular function in our new graph, and we give them informal names that reflect these functions: we call $V_1$ the \emph{vertex-representatives}, $V_2$ the \emph{edge-representatives}, and $V_3$ are the \emph{brothers} of vertices in $V_1$ (and in particular, each vertex in $V_3$ is the brother of the member of $V_1$ it is related to by $f$.

Then let $E'' = E_1 \cup E_2 \cup E_3 \cup E_4 \cup E_5 \cup E_6$ such that :
\begin{my_enumerate}
\item $E_1  = \{(v_i, v_j) | v_j \in E'$ and there exists $v_k$ such that
              $v_j = (v_i, v_k)\}$
\item $E_2 = \{(v_i, v_j) | v_i, v_j \in V_3 \cup V_2 \}$
\item $E_3 = \{(v_i, v_j) | v_j = f(v_i)\}$
\item $E_4 = \{(v_i, v_s) | v_i \in (V_3 \cup V_2)\}$
\item $E_5 = \{(v_i, v_b) | v_i \in (V_3 \cup V_2)\}$
\item $E_6 = $ the edges of $G_d^u$
\end{my_enumerate}
Let $G'' = (V'', E'')$.  Due to symmetry in the graph, we can label $v_s, v_b, v_{s'}$ and $v_{b'}$ in $V_4$ as in Figure \ref{fig:GDU}.
We call $G''$ the \emph{$G_d^u$-blocked graph} of $G$, and $G_d^u$ the \emph{blocker} of $G''$.

The main idea of our reductions is to block off parts of the underlying host tree, forcing representation of vertices corresponding to vertices in the original graph onto a limited number of twigs.  In the case of our reductions on sub\textbf{tree} overlap graphs (as opposed to the reductions on sub\textbf{path} overlap graphs), this is achieved by the use of the $G_d^u$ blocker, which will occupy all nodes in the host tree of degree greater than three (that is, branching nodes).   The brothers of the vertex-representatives will force those vertex-representatives to be disjoint in the representation.  The edge-representatives will be un-representable if the original graph is not colourable.  These three sets of vertices in combination will force the vertex-representatives to be placed only on the twigs of the host tree.  

Despite our careful construction, it will be possible for a strictly limited number of vertex-representative vertices to be placed in a  way that does not correspond properly to a colouring in the original graph (we will later describe these in detail as \textit{illegal pairs}).  To overcome this difficulty, we include more disjoint copies of the original graph in our reduction than can exist illegal pairs, guaranteeing that if the overall graph is represented, then at least one set of vertex-representatives is represented in a way that gives us a colouring of the original graph.  

The main idea in the case of our reductions on sub\textbf{path} overlap graphs is similar to that for the sub\textbf{tree} overlap graphs, but because of the more restricted representation, we do not need to include $G_d^u$ in the instance we produce: a graph with only the vertex-representatives, edge-representatives, and brothers will suffice.  

 Figure~\ref{fig:gdup2} demonstrates how to imagine graph $G''$ and sets $V_1, V_2$ and $V_3$, i.e., vertex-representatives, edge-representatives and brothers of vertex-representatives, respectively. They are depicted using individual colours. Vertex-representatives are yellow, edge-representatives blue and brothers peach. Note that peach and blue vertices are forming a clique. This graph we assign to $P_2$.
\begin{figure}
\centerline{\scalebox{0.7}{\includegraphics{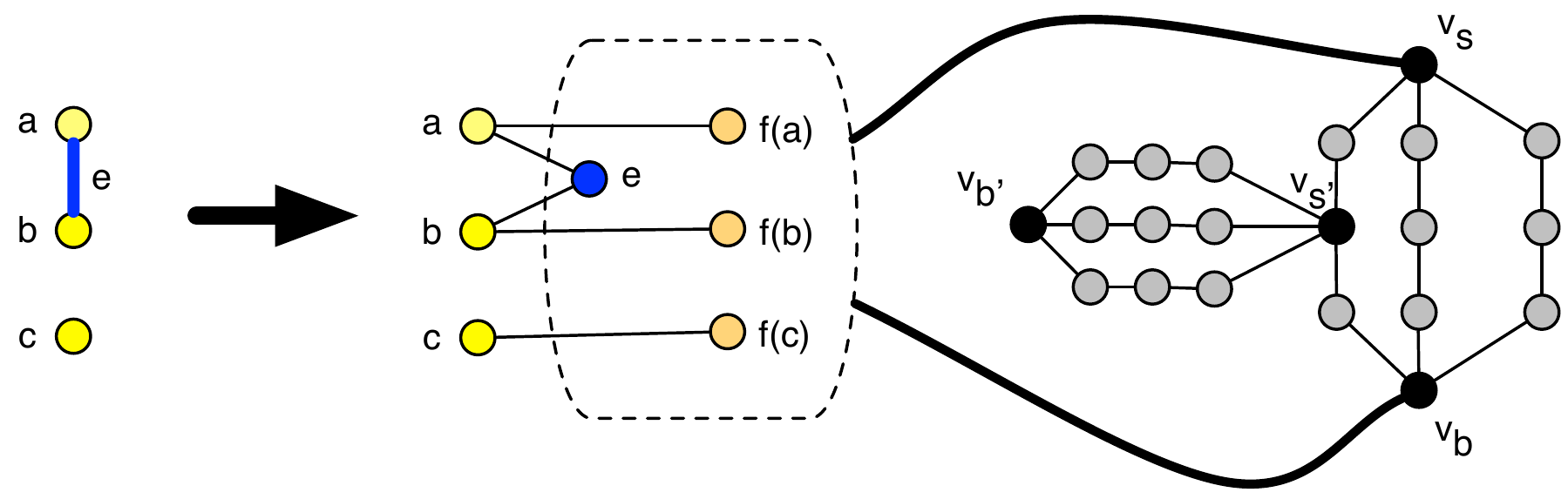}}}
\caption{Demonstration of constructing $G''$ (on the right side of the arrow) from $G$ (on the left side of the arrow). Yellow vertices ($a, b, c$)  are vertex-representatives, peach vertices ($f(a), f(b), f(c)$) are their brothers and blue vertices are edge-representatives. Note that all vertices in the dotted oval induce a clique, and are all adjacent to $v_b$ and $v_s$.  In a full construction, there will be six copies of $G$ involved, and therefore six copies of the vertex-representatives, brothers of vertices, and edge-representative.  All edge-representative and brothers will induce a clique.  These copies have been omitted for legibility.  }
\label{fig:gdup2}
\end{figure}

Because of the similar ideas behind all the reductions, we proceed to describe some of the details first, before launching into the reductions themselves.  First, we describe the special cases of illegal representations of vertex-representative, and show that in any representation of a $G_d^u$-blocked graph, there are a limited number of these special cases.   This illegality exists with respect to a particular representation of the $G_d^u$-blocked graph, and therefore the definition includes both graph and representation:

\begin{definition}
Let $G'' = (V'', E'')$ be the $G_d^u$-blocked graph of $G = (V, E)$, represented on host tree $T$, with vertex classes $V_1, V_2, V_3, V_4$ as defined above, and let $v_i, v_j \in V_1$ such that $(v_i, v_j) \in E$.  Vertices $v_i, v_j$ are an \textbf{illegal pair} if  $t_i, t_j$ intersect the same twig of $T$.  
\end{definition}

When discussing illegal pairs, we will also sometimes refer to the subtrees or subpaths representing illegal pairs of vertices as illegal themselves.

\begin{figure}
  \begin{center}
\scalebox{0.35}{\includegraphics{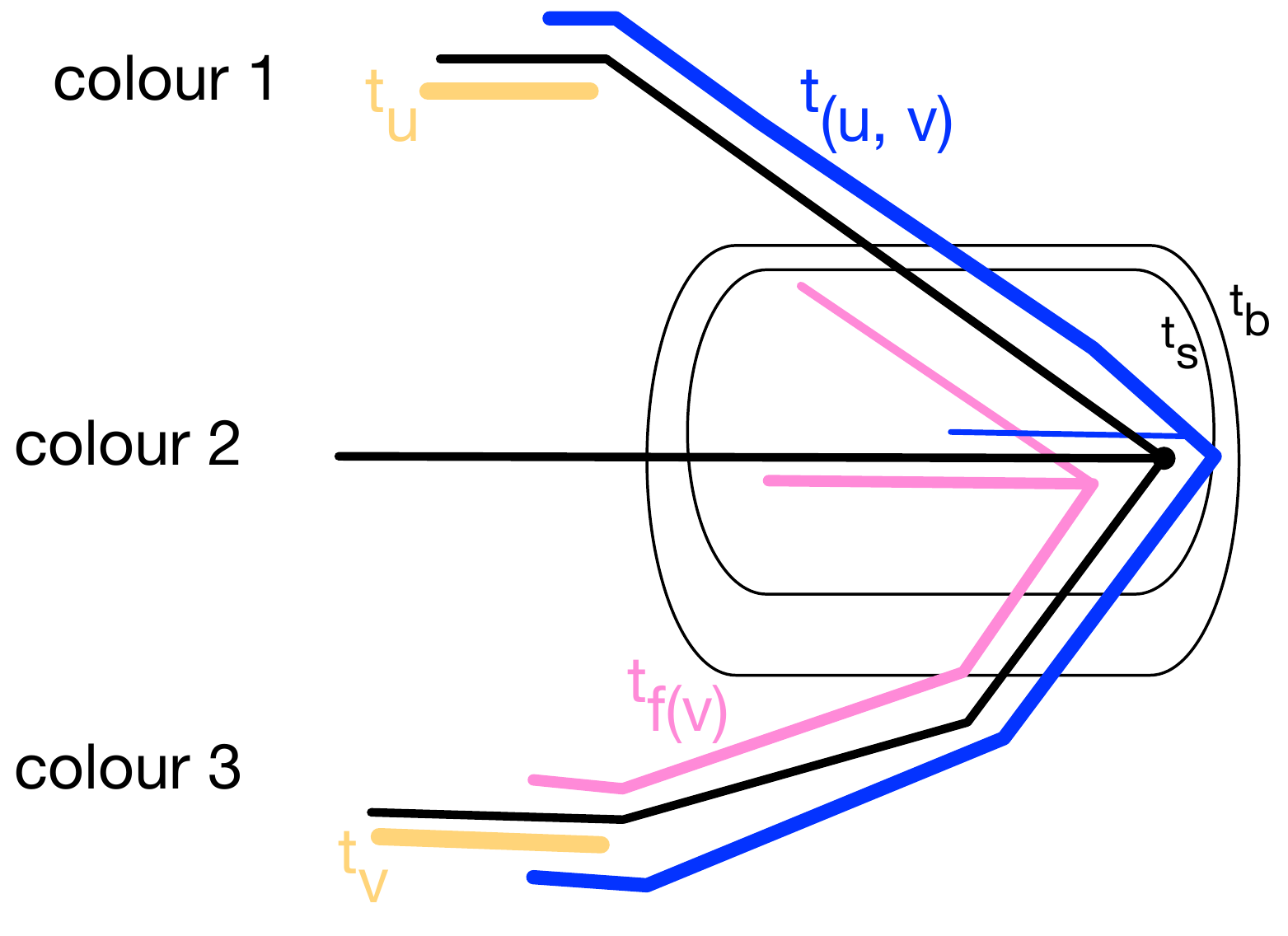}}
      \caption{A quick intuition on the representation of vertices in $V_1, V_2, V_3$ in a $G_d^u$-blocked graph on a tree with a single vertex of degree greater than two.  Representing the $G_d^u$-blocked graph of a 3-colourable graph $G$:  vertex $v$ is a copy in the $G_d^u$-blocked graph of a vertex that is coloured with colour 1 in $G$, and vertex $u$ is a copy in the $G_d^u$-blocked graph of a vertex that is coloured with colour 3 in $G$.  Vertices $u$ and $v$ are adjacent in $G$.  Recall that there are six copies of the vertex set of $G$ in the vertex set 
of the $G_d^u$-blocked graph of $G$ - we shall not consider these explicitly here.  If the underlying tree is shown in black, then we assign each twig to a colour.  We represent $t_u$ on the twig for colour 1, and $t_v$ on the twig for colour 3.  Shown in blue, $t_{(u, v)}$ is the subtree for the edge-representatives that represents the edge between $u$ and $v$ in $G$: it has an endpoint in each of $t_u$ and $t_v$.  We show in pink the subtree $t_{f(v)}$, the subtree representing the brother of $v$.  Blocking subtrees $t_s$ and $t_b$ contain the central vertex, and are contain all of the tree within their sketched outline ovals.  The force all edge-representatives and brothers to overlap, short paths of appropriate length are added to them, branching from the central vertex (but still within $t_s$).  
} \label{fig:exRep}
  \end{center}
\end{figure}

We will make extensive use of several technical properties of representations of $G_d^u$ (and therefore $G_d^u$-blocked graphs): we concern ourselves with these properties for the remainder of this section.

Rosgen \cite{rosgenThesis} provides us with the following useful Lemma:
\begin{lemma}\label{lem:Rosgen3.14}
Let graph $G = (V, E)$ be represented by subtrees $\mathcal{T}$ of tree $T$.  Let $v_i, v_j$ be non-adjacent vertices in $V$.  If $t_i \subset t_j$ then for every vertex $v_k$ such that there exists a path of $G$ from $v_i$ to $v_k$ that does not intersect the neighborhood of $v_j$, it holds that $t_k \subset t_j$.
\end{lemma}

By applying this lemma, Enright \cite{Jess} showed that in any representation of $G_d^0$ where $d$ is at least 3, either $t_s \subset t_b$ or $t_b \subset t_s$: because $G_d^0$ is an induced subgraph of every $G_d^u$, this also holds in all representations of $G_d^u$.  Because of the symmetry in our construction, we can assume, without loss of generality that $t_s \subset t_b$: note also that the subscripts $s$ and $b$ are mnemonics for \textit{small} and \textit{big}, respectively. 

By further applying Rosgen's lemma on our $G_d^u$, we obtain:

\begin{lemma}\label{l:tstbcontained}
Without loss of generality, in any subtree-overlap representation of $G_d^0$ with $d\ge 3$, $t_s\subset t_b$. In any subtree-overlap representation of $G_d^u$ with $d\ge 3$ and $u\ge 2$, $t_s\subset t_b$ and $t_{s'}\subset t_{b'}$.
\end{lemma}
\begin{proof}
We include here a sketch of our use of Rosgen's lemma.  As above, we can assume that $t_s \subset t_b$ without loss of generality.  Because the graph composed of $t_{s'}$ and $t_{b'}$ and the paths between them is isomorphic to $G_3^0$, we know that either $t_{s'} \subset t_{b'}$ or $t_{b'} \subset t_{s'}$.  If $t_{b'} \subset t_{s'}$, then by Rosgen's lemma we have that $t_s$ and $t_b$ and the paths between them that do not include $t_{b'}$ are also subtrees of $t_{s'}$.  However, because there is a path from $t_{s'}$ to $t_{s}$ that avoids the neighbourhood of $t_b$, we also have that $t_{s'}$ is a subtree of $t_s$, a contradiction.




\end{proof}

Because in our reductions we consider classes of host trees that are closed under edge subdivision, we can assume without loss of generality that in any representation of $G_d^u$ in which any of $t_s, t_b, t_{s'}$ and $t_{b'}$ contain a branching-node, it contains also all its neighbours.  Note also that subdividing the edges of any tree cannot increase the number of leaves of that tree.  We will make use of these facts in the upcoming lemmas.  

\begin{lemma}\label{l:spanbranch}
Let $T$ be a tree with $k$ leaves, and $s$ and $t$ be disjoint subtree of $T$ with $l$ and $k-l+2$ boundary nodes, respectively.  Then all nodes of $T$ of degree greater than two are either in $s$ or in $t$, and this remains true in any subdivision of the representation.
\end{lemma}
\begin{proof}
We proceed by contradiction: assume there is a branching node $b$ that is not in $s$ or $t$.  If $b$ is on the unique path between $s$ and $t$, then consider the subtree formed by $s$, $t$, the path between them (which includes $b$), and at least one neighbour of $b$ that is not on that path (guaranteed to exist because $b$ is of degree greater than 2).  
Now consider the number of leaves of this subtree:  all but one of the leaves of $s$ is a leaf of this new subtree, all but one of the leaves of $t$ is a leaf of this new subtree, and one neighbour of $b$ is a leaf of the new subtree.  Then the subtree has $l + k -l + 2 -2 + 1 = k+1$ leaves - impossible for a subtree of $T$, which has $k$ leaves. 

If $b$ is not on the path between $s$ and $t$ that argument proceed similarly: we produce a subtree that includes $s$, $t$, and $b$ and all of its neighbours, and the paths between these, producing a subtree with $k+1$ vertices - a contradiction.

This result holds in any subdivision of the representation because subdivision of any tree does not affect its leafage. 

\end{proof}
Also note that the proof of Lemma~\ref{l:tstbcontained} gives us as a corollary the following useful statement:
\begin{lemma}\label{l:tstbdleaves}
In any subtree-overlap representation of $G_d^0$, both, $t_s$ and $t_b$ have at least $d$ leaves.
\end{lemma}

\section{Subtrees in restricted trees}

In this section we describe the use of the auxiliary graphs constructed in the previous section to show hardness of recognition for graphs represented by subtrees in restricted trees. The first of our theorems is the following:
\begin{theorem}\label{t:fixedleaves}
REC-T-$k(G)$ is NP-complete.
\end{theorem}

In order to show this theorem, for a given $k$ we reduce 3-CON-$k$-COL$(G)$ to REC-T-$k(G)$: because the case in which $k=3$ is notably simpler, we start with that case as an example of the main ideas.  We will then move into the more general case, which has significantly more required technical detail.  After finishing with  REC-T-$k(G)$, we will turn our attention to a similar problem with a similar reduction: REC-SUB-T$(G)$.  

For a given 3-connected graph $G$ (instance of 3-CON-$k$-COL$(G)$), we define a graph $G''$ as the $G_d^u$-blocked graph $G$, where $d=3$ and for $k=3$, $u=0$, otherwise $u=k-d+1$. In the rest of this section we show that $G$ is $k$-colourable if and only if $G''$ has subtree overlap representation on some tree with $k$ leaves.

As we have outlined previously, the common idea of these reductions is to block-off the central part of the host tree in any representation, requiring representation of vertices in target instance corresponding to vertices in the colouring instance to be represented on $k$ twigs of the host tree.  The representation is engineered such that vertices represented on twigs in this way correspond to acceptable colour classes in the colouring instance: this is achieved using the edge-representative vertices, which can only be represented if their endpoints in the colouring instance are on different twigs: thus two vertex-representatives in the representation can only be on the same twig if they are non-adjacent in the source colouring instance.  

Unfortunately, there are a number of ways in which this representation can "go wrong" with respect to corresponding to acceptable colour classes.  Showing that these possible problems are limited in number, and that therefore working with multiple copies of the original instance can surmount this difficulty contributes most of the technical difficulty in the upcoming proofs.  

Most of these details are not needed in the much-simpler case of $k=3$, and we therefore use this case as a technical warm-up, with the aim of making the core ideas clear before we embark on the more detailed, but essentially similar, case for general $k > 3$.

\subsection{Special case: REC-T-$3(G)$}
Because we will re-use parts of this proof in upcoming reductions, we split the statement into two separate lemmas, first showing that if $G$ is 3-colourable, then the $G_3^0$-blocked graph of $G$ is an overlap graph of subtrees in a tree with 3 leaves, and then showing the reverse implication. 

\begin{lemma}
Given a 3-connected graph $G$, if $G$ is 3-colourable, then the $G_3^0$-blocked graph of $G$ is an overlap graph of subtrees in a tree with 3 leaves.
\end{lemma}
\begin{proof}

First, we show that if $G$ is 3-colourable, then the $G_3^0$-blocked graph of $G$ is an overlap graph of subtrees in a tree with 3 leaves. We start by representing the vertices of $G_3^0$ with $t_s \subset t_b$ and the node of degree three in them both.  The three paths between them are represented along the three paths from that branching node, as in Figure~\ref{fig:exRep}.  

Let $C_1, C_2, C_3$ be the three colour classes of a 3-colouring of $G$.  For convenience, we label the twigs of our host tree as twigs 1, 2, and 3.  In any arbitrary order along the twig, we represent the vertex-representatives corresponding to vertices in $G$ by two-node subtrees, with a vertex-representative corresponding to a vertex in $C_i$ on twig $i$.  Thus far, the vertex-representatives are an independent set, and the adjacencies between vertices in $G_3^0$ are correct.  We must now represent the edge-representatives and the brothers of the vertex-representatives.   We will describe their placement such that they have the correct overlapping with vertex-representatives and vertices of $G_3^0$, and then describe an alteration to them that allows them to all pairwise overlap, as required.  

We assign to each edge-representative $e=(u, v)$ the unique path overlapping both the subtree for $u$ and the subtree for $v$ (already represented on different twigs, as they are vertex-representatives and cannot be in the same colour class).  We assign to each brother $f(u)$ the unique path that has as its leaves the branching node and the node of the subtree representing $u$ that is closest to that branching node.  Observe that each brother overlaps only the vertex-representative that it ought to, each edge-representative overlaps only the two vertex-representatives that it ought to, and both the brothers and the edge-representatives overlap exactly $t_s$ and $t_b$ of the vertices of $G_3^0$.  

However, the edge-representatives and the brothers are not yet guaranteed to all pairwise overlap: we remedy this with a small adjustment.  

If $N$ is the number of edge-representatives and brothers, then we add to the host three paths of length $N$  by subdividing the edges between the branching node and its neighbours - for convenience, call these paths $P_1, P_2, P_3$.  We add these paths also to both $t_s$ and $t_b$, but not yet to any other subtrees - note that no other subtrees contained this subdivided edge.  We will use these paths to produce pairwise overlapping within the set of edge-representatives and brothers.  For $P_1$, we take the set of subtrees representing brothers and edge-representatives that have endpoints on the twigs containing $P_2, P_3$.  We sort these by size, breaking ties arbitrarily (note that two subtrees with the same size already overlap, as no two are equal), and extend these subtrees out along $P_1$, extending the previously-smallest the farthest, and the previously-biggest the least, such that no two share a leaf on $P_1$.  We repeat this operation on $P_2, P_3$.  As a result, all subtrees representing brothers and edge-representatives now overlap, and no other overlapping has been impacted.  This is then a representation of the $G_3^0$-blocked graph of $G$ on a host tree with three leaves.  

A schematic of this complete representation can be found in Figure~\ref{fig:exRep}.

\end{proof}

The essential form of the previous lemma's proof will be repeated, but with more required technical details, in all of the upcoming reductions. 

We must now prove the other direction, and show that if the $G_3^0$-blocked graph of $G$ is an overlap graph of subtrees in a tree with 3 leaves, then $G$ is 3-colourable.  We state it in a slightly more general way, to facilitate its reuse later.

\begin{lemma}
\label{l:previous}
Let $G = (V, E)$ be a 3-connected graph, and $G''$ be the $G_d^0$-blocked graph of $G$.  If $G''$ is the overlap graph of subtrees $\mathcal{T}$ of a tree $T$ with $k = d$ leaves and no node of degree greater than two but less than $d$, then $G$ is $k$-colourable.
\end{lemma}
\begin{proof}

Let $G = (V, E)$, $G'' = (V'', E'')$, $T$, $\mathcal{T}$, $d$ and $k$ be as described in the Lemma statement. We will show that any representation must be reasonably close to a nice one (in the sense of Observation \ref{obs:niceRepImpliesColouring}). 

More precisely, we will show that at most four vertex-representative subtrees deviate from a nice representation.  By the definition of a nice representation, if a vertex-representative deviates from a nice-representation, then either:
\begin{enumerate}
\item it is a member of an illegal pair (i.e., representatives of two neighboring vertices in $G$ are on the same twig in the representation of $G''$), or
\item it is not contained in a twig, and therefore contains the branching node.
\end{enumerate}

From Lemma~\ref{l:tstbcontained} we may conclude that w.l.o.g., $t_s\subset t_b$, and from Lemma~\ref{l:tstbdleaves} we see that both $t_s$ and $t_b$ must contain the branching vertex and (because we can subdivide representations without changing their overlap graph) we may assume that $t_s, t_b$ both also contain all neighbours of that branching node.

Recall that in constructing $G''$ we included six disjoint copies of the vertex set of $G$. To show that there is at least one copy (of the vertices of the original graph) nicely represented, we observe the following:
\begin{enumerate}
\item At most one vertex-representative can be represented as a supertree of $t_b$.  
\begin{itemize}
  \item 
  We proceed by contradiction: assume that $u, v$ are vertex representatives such that $t_b \subset t_u$ and $t_b \subset t_v$,  without loss of generality, assume that $t_u \subset t_v$ (they cannot overlap, as vertex-representatives are an independent set in $G''$).  Then the brother of $t_v$ cannot be represented: it must overlap $t_b$ and $t_v$, but not $t_u$; an impossibility because $t_b \subset t_u \subset t_v$.  
\end{itemize}
\item At most one vertex-representative can be represented as a subtree of $t_s$ containing the branching node.  
\begin{itemize}
  \item 
  We proceed by contradiction: assume that $u, v$ are vertex representatives such that (without loss of generality) $t_u \subset t_v \subset t_s$ and $t_u, t_v$ contain the branching node.  Then the brother of $t_u$ cannot be represented: it must overlap $t_s$ and $t_u$, but not $t_v$; an impossibility because $t_u \subset t_v \subset t_s$.  
\end{itemize}
\item For two vertex-representatives $u, v$ such that, both $t_u, t_v$ are either disjoint from $t_s$ or contained in $t_s$, it holds that $t_u$ and $t_v$ are disjoint.
\begin {itemize}
 \item If $t_u$ and $t_v$ are both disjoint from $t_s$, then they are on the twigs of the representation.  We proceed by contradiction: if they are not disjoint, then without loss of generality $t_u \subset t_v$.  The brother of $t_u$ then overlaps $t_v$, as it intersects $t_u$ (and therefore $t_v$), does not contain all nodes of $t_u$ (and therefore of $t_v$), and contains at least some nodes in $t_s$ (which are not in $t_v$), a contradiction.  
\end{itemize}

\item Given an illegal pair, one member of the pair must be represented as a subtree of $t_s$ while the other is disjoint from $t_b$.
\begin {itemize}
 \item Let $t_u, t_v$ be the illegal pair - recall that this means that $u, v$ are adjacent in $G$, and $t_u, t_v$ are on the same twig of the representation.  The edge-representative of $(u, v)$ must overlap both $t_s$ and $t_b$, as well as $t_u$ and $t_v$.  If both $t_u, t_v$ are disjoint from $t_s$, then the edge representative is confined to the path between them and is disjoint from $t_s$.   If both are subtrees of $t_b$ then the edge representative is confined to the path between them and is a subtree of $t_b$.  
\end{itemize}

\item No vertex-representative is a subtree of $t_b$ but disjoint from $t_s$.
\begin {itemize}
 \item 
  Let $t_u$ be a vertex-representative subtree that is a subtree of $t_b$.  Then, by Rosgen's Lemma, $t_u$ must not be a supertree of any subtree representing a vertex in $G_3^0$. Here we consider two cases of contradiction, showing that the brother of the vertex-representative restricts us in both cases. 
 
  If $t_u$ intersects any subtree representing a vertex in $G_3^0$ other than $t_s, t_b$ it must therefore be contained in it, and the brother subtree $t_{f(v)}$  could not overlap $t_s$ and $t_b$, but not $t_x$, as required.  
 
 If $t_u$ does not intersect any of the subtrees representing vertices in $G_3^0$ other than $t_b$, there must be nodes in $t_b$ that are not in any of those subtrees.  Enright \cite{Jess} showed that in any representation of $G_3^0$ on a host tree with one vertex of degree three in which there are nodes in $t_b$ that are not in any other subtree, one of the subtrees (call it $t_x$) representing a neighbour of $v_b$ must contain $t_s$.  Then the brother subtree $t_{f(v)}$  could not overlap $t_s$ and $t_b$, but not $t_x$, as required.  
 \end{itemize}

\item Illegal pairs must not be nested, i.e., given a pair of illegal pairs $t_u,t_v$ and $t_w,t_x$ (on the same twig), they cannot be represented (along this twig) in the ordering $t_u, t_w, t_x, t_v$. 
  \begin{itemize}
    \item  If this happened, the edge-representatives (overlapping $t_u, t_v$ and $t_w, t_x$, respectively) could not overlap.
  \end{itemize}

\item Two illegal pairs cannot be represented on two distinct twigs.
  \begin{itemize}
    \item  The edge-representatives of the two illegal pairs must intersect: this is impossible if they are confined to the disjoint paths between the illegal pairs on distinct twigs. 
  \end{itemize}

\end{enumerate}

Note that by the first, second, and fifth observations above, we know that at most two subtrees corresponding to vertex-representatives may contain the branching vertex - all other vertex-representatives must to be represented on twigs.  In addition, the seventh observation tells us that illegal pairs may occur on only one twig.  We will show that there can only be one illegal pair in total, and, combined with the limited number of vertex-representatives that may contain the branching vertex, this means that at most three of the copies of the vertex set if $G$ are not nicely represented.  Therefore at leat three \textbf{are} nicely represented, and therefore $G$ is 3-colourable.  Of course, we only need one nicely represented copy, so in this case we could have included fewer copies: in a future reduction we will require all six. 

Consider a one copy of the vertices of $G$ in the set of vertex-representatices of $G''$ - we will call this single copy $V_x$.  We show that there is at most one illegal pair in $V_x$.  Recall that $G$ is 3-connected - we will finally make use of this fact!

We will first show that it is not possible to represent all of the vertices in $V_x$ as illegal pairs.  
By our fourth observation above, we know that of every illegal pair, one member must be a subtree of $t_s$, and the other must be disjoint form $t_b$.    Consider a triangle of vertices in $G$: if every edge in this triangle has an illegal pair as its endpoints, then one of the endpoints must be a subtree of $t_s$ and the other disjoint from $t_b$, an impossibility by the pigeonhole principle, therefore no triangle in $G$ can be represented entirely as illegal pairs in the representation of $G''$.  Because the only 3-connected graph on four vertices is $K_4$, we know that if $G$ has four vertices, $V_x$ is not represented entirely by illegal pairs.  If $G$ has at least five vertices, then we proceed by contradiction.  Assume that all vertices in $V_x$ are represented on the same twig (that is, every adjacent pair is an illegal pair).  Let $t_u$ be the vertex-representative that is closest to the branching node - note that $t_u$ must be a subtree of $t_s$.  Because $G$ is 3-connected, $u$ must have at least three neighbours $v, w, x$ in $G$.  All three are in illegal pairs with $u$, and so $t_v, t_w, t_x$ are all disjoint from $t_b$ along this same twig - say that $t_v$ is the closest to the branching node and $t_x$ closest to the leaf (all three must be disjoint, by our third numbered observation).  Again, because $G$ is 3-connected, $v$ must have another neighbour $z$, and it must each form an illegal pair with $v$, and therefore $t_z$ must be in $t_s$, but $t_z$ cannot be closer to the branching node than $t_u$ (by the choice of $t_u$) and must be disjoint from $t_u$ (by our third numbered observation), and therefore the illegal pair $t_v, t_z$ is nested between the illegal pari $t_u, t_x$, a contradiction to our sixth numbered observation!  We can conclude that not all adjacent pairs of vertices in $V_x$ form illegal pairs in the representation.

Now that we know there is at least one legal pair in $V_x$, we shall exploit this fact to show there is at most one illegal pair.  We proceed by contradiction:
assume that there are at least two illegal pairs, and consider all illegal pairs on a twig: one member of each is a subtree of $t_s$, and one is disjoint from $t_b$.  Then when we sweep the twig from the branching node (designated the "left") toward the leaf (designated the "right"), we encounter (before the nodes in $t_b$ that not in $t_s$) one-half of each illegal pair, and then (after the nodes in $t_b$ that not in $t_s$) the other half of each pair.  We will informally describe these as the "left" endpoint (toward the branching node) and the "right" endpoint (toward the leaf) of each illegal pair.  Of course, some subtrees may be the left endpoint of multiple illegal pairs, or the right endpoint of multiple illegal pairs.  Let $t_u$ be the rightmost left endpoint, and $t_v$ be the leftmost right endpoint.  Any vertex-representative to the left of $t_u$ or the right of $t_v$ cannot have a legal neighbour, as the edge-representative for that legal edge would be either disjoint from the edge-representative of an illegal pair to the right, or contain the edge-representative of an illegal pair to the left.  Therefore at most two vertices of all the vertices in illegal pairs may have legal neighbours, and we know that at least one vertex has a legal neighbour.  Consider then the two vertex subsets of $V_x$: those represented on the twig we have been considering (at least three, as we have at least two illegal pairs), and those represented elsewhere (of which we know there is at least one).  Our (at most) two vertices on this twig with legal neighbours elsewhere are a cutset between these two vertex subsets: a contradiction to $G$ being 3-connected.

Therefore at most one pair of vertices in $V_x$ that are adjacent in $G$ can be an illegal pair in the representation of $G_3^0$, completing the proof of the lemma.  
 \end{proof}
 
The two immediately previous lemmas together suffice to show our Theorem~\ref{t:fixedleaves} for $k=3$, i.e., subtrees in trees with 3 leaves.  Indeed, both are readily adaptable to cases of stars with larger numbers of leaves.  However, for leafage greater than three, there are more complex trees with more than one branching node, and we must adapt our reduction to block all of these branching vertices and everything between them.  This will be the work of the number of paths specified by the $u$ parameter in our $G_d^u$ graph, which we have not yet made use of.  

\subsection{The more general case} 
Because we will re-use parts of the proof from our subsequent proof on the hardness of REC-SUB-T($G$), we formulate portions of it as a lemma in the more general case:

\begin{lemma}
\label{lem:moreGeneralList}
Given a subtree overlap representation of $G''$, a $G_d^u$-blocked graph $G$, in a tree $T$ such that all branching-vertices (of $T$) are contained in $t_b$, then:
\begin{itemize}
\item At most one vertex-representative can be represented by a tree containing $t_b$ and all other vertex-representatives must be pairwise disjoint.
\item Illegal pairs must not be nested.
\item When restricting only on vertex-representatives represented on twigs, all illegal pairs must lie on the same twig.
\item The whole subgraph of $G''$ corresponding to one complete copy of $G$ cannot be represented only by illegal pairs.
\item When restricting only on vertex-representatives represented on twigs that take part on some illegal pair, except of (at most) two of them, all other must have all neighbors represented by illegal pairs.
\end{itemize}
\end{lemma}
\begin{proof}
The reasoning here is essentially the same as in the case with only three leaves, but we provide quick sketches or reminders for each:
\begin{itemize}
\item 
At most one vertex-representative can be represented by a tree containing $t_b$,  and all pairs of vertex-representatives that are both disjoint from $t_b$ or both subtrees of $t_b$ must be themselves disjoint.
   \begin{itemize}
      \item Proof: Here we use the fact that, given three subtrees $t_1 \subset t_2 \subset t_3$, any fourth subtree that overlaps both $t_1$ and $t_3$ is also forced to overlap $t_2$.  
      If two vertex-representative subtrees contain $t_b$, then the brother of the larger is unrepresentable (as it must overlap both $t_b$ and the larger vertex-representative subtree, but not the subtree between them).  If two vertex-representatives are non-disjoint and are both subtrees of $t_b$, then the brother of the smaller one is unrepresentable.  Using similar reasoning, if two vertex-representatives are disjoint from $t_b$ but one is a subtree of the other, then the brother of the smaller one is unrepresentable.
   \end{itemize}
\item Illegal pairs must not be nested (as defined within Lemma~\ref{l:previous}).
 \begin{itemize}
\item Proof: As in Lemma \ref{l:previous}, if illegal pairs are nested, then the edge-representative of one will be confined to a subtree of the edge-representative of the other, impossible in a valid representation in which these must overlap.
  \end{itemize}
\item When restricting only on vertex-representatives represented on twigs, all illegal pairs must lie on the same twig.
 \begin{itemize}
\item Proof: As in Lemma \ref{l:previous}, if illegal pairs as on different twigs, their edge-representatives are restricted to subpaths of those two twigs, and therefore must be disjoint - a contradiction.  
  \end{itemize}

\item There is at most one illegal pair in the representation.
 \begin{itemize}
\item Proof: As in Lemma \ref{l:previous}, this follows from the 3-connectedness of $G$.
  \end{itemize}
\end{itemize}

\end{proof}

As we turn our attention to the more complex case of host trees of leafage more than three, our approach parallels that for the simple case:  we show that in any representation of the appropriate $G_d^u$-blocked graph, there are a strictly limited number of vertex-representatives that are not nicely represented, and that this number is strictly less than the number of copies of the vertex set of $G$ that we have included.  

To shorten our technical presentation, we present as a single lemma the result that almost immediately implies our two main theorems on recognising the overlap graphs of host trees of limited leafage, and on recognising the overlap graphs of host trees derived from a n input tree by subdivision.  

We will use the notion of a \emph{lastbranch} node; recall from earlier in this manuscript that, given a tree $T$, a lastbranch of $T$ is a node that has degree $d \geq 3$, and is incident to at least $d-1$ twigs.  Informally, a lastbranch node is the "last" branching node in the direction of all but one of its neighbours when moving out towards the leaves of the tree.   Removing a lastbranch disconnects the tree into at most one non-path tree and a collection of paths.  
While we have so far dealt mainly with host trees that have a single, unique lastbranch (the single node of degree greater than two), trees may have many lastbranches, and it is to these that we now turn.  We will show that appropriate $G_d^u$-blockers may be used to cover all lastbranches of a tree, the nodes on paths between them, leaving only twigs as the hosts for most vertex-representatives.  

\begin{lemma}\label{l:morebranch}
Let $G = (V, E)$ be a 3-connected graph, $T$ a tree with $k>3$ leaves and no lastbranch node of degree smaller than $d$.  Let $G''$ be the $G_d^{k-d+1}$-blocked graph of $G$.  Then $G''$ is the overlap graph of subtrees $\mathcal{T}$ of some subdivision of $T$, if and only if $G$ is $k$-colourable.
\end{lemma}
\begin{proof}
 As in the case of the host tree with three leaves, we show each direction of this implication separately.  
 
 Starting with a 3-connected graph with a given $k$-colouring, we define a canonical representation of the $G_d^{k-d+1}$-blocked graph of $G$ in which twigs represent colour classes.  In this representation $G_d^u$ gets represented over the branching nodes in such a way that $t_s$ spans the lastbranch of the smallest degree (i.e., $d$) and $t_{s'}$ and $t_{b'}$ span all other branching nodes. Vertex-representatives are represented exactly as in the previous case, i.e., the colour determines the twig they are represented on. Their brothers and edge-representatives are represented in a similar way (though they span all branching nodes instead of just the unique one), and again are adjusted to force pairwise overlapping as in the previous case. Figure~\ref{fig:repGDU} gives the intuition for this representation: we do not give full written details because it is so similar to the fully-described case for the host tree with three leaves.
\begin{figure}
\scalebox{0.35}{\includegraphics{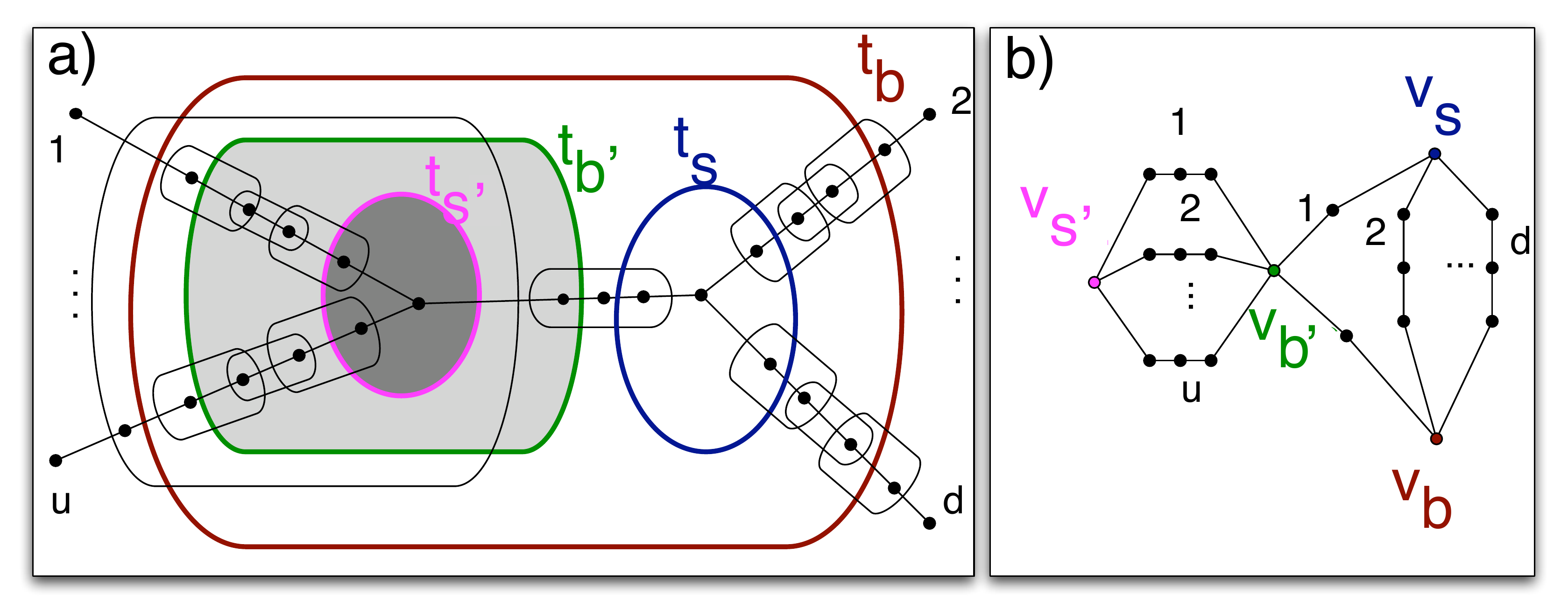}}
\caption{
A generalised overlap representation of the $G_d^u$ graph on a tree with a node such that the forest created by removing that node has two connected components:  a tree with $d$ leaves and a node of degree $d$ and a tree with at least $u + 1$ leaves.  The interior of $t_{b'}$ and $t_{s'}$ are darkened to indicate that the structure of the tree there is somewhat irrelevant - only the number of boundary nodes is important.  There is exactly one node of degree greater than two contained in $t_s$ and that node is contained only in $t_s$ and $t_b$, and  all other nodes of degree greater than two are contained in $t_{b'}$.  The representation is on the left, and the $G_d^u$ graph is on the right.  Vertex labels and corresponding subtrees are colour coded.}
\label{fig:repGDU}
\end{figure}

Proof of the converse must show that all representations are almost as nice as the canonical one - again we will show that, in any representation, there are a limited number of vertex-representatives that are not nicely represented, and therefore at least one copy of the vertex set of $G$ is nicely represented, giving a $k$-colouring for $G$.

First, we show that $t_s, t_{s'},$ and $t_{b'}$ are all subtrees of $t_b$, that $t_{s'}\subset t_{b'}$, and that $t_s$ is disjoint from both $t_{s'}$ and $t_{b'}$.  As the corresponding vertices form an independent set in $G''$, we know that these subtrees are either disjoint or contained one in another.  By Lemma \ref{l:tstbcontained}, we can assume that $t_s \subset t_b$ and $t_{s'} \subset t_{b'}$.  Combining this with the existence of a path between $t_{b'}$ and $t_s$ that avoids the neighbourhood of $t_b$, and an additional application of Rosgen's Lemma, we have that $t_s, t_{s'},$ and $t_{b'}$ are all subtrees of $t_b$.    Again, by Rosgen's Lemma, if $t_{b'}$ were a subtree of $t_s$, then $t_b$ would also be: a contradiction - therefore $t_{b'}$ (and thus also its subtree $t_{s'}$) are both disjoint from $t_s$.

 Lemma~\ref{l:tstbdleaves} gives us bounds on number of leaves of $t_s, t_b, t_{s'}$ and $t_{b'}$. As $t_{b'}$ and $t_s$ together have $k+2$ leaves, by Lemma~\ref{l:spanbranch} they span all branching-nodes. We know that $t_s$ spans at least one branching node and it has to be a lastbranch (as $t_s$ and $t_{b'}$ are disjoint). The lastbranch of minimum degree suffices for $t_s$ while all other branching nodes are covered by $t_{b'}$.
 
Now we are ready to realize that, similarly to the previous case of one branching node, at most one vertex-representative may contain $t_b$, and at most one vertex-representative may be a subtree of $t_s$ that contains the branching node (recall there is only one in $t_s$) - both are implied by Lemma~\ref{lem:moreGeneralList}.  
 
But what about the rest of the non-twig interior of $t_s$?  While this consisted only of the branching node in the 3-leaf case, here there is also a path from the branching node toward the nodes of $v_{b'}$.  We show that we can adjust any representation so that only one vertex-representative is on this path in addition to the vertex-representative containing the branching node (which will combine with the single vertex-representative that may contain the branching node in $t_s$ to tell us that we can assume at most two vertex-representatives that are a subset of $t_s$).  

Our argument for adjusting the representation hinges on the fact that if there is twig on which there are no vertex-representatives that are neighbours in $G$ of a vertex-representative $u$ we are working with, then we can move $t_u$ to the leaf-side end of that twig (extending it if needed), and adjust the brother and edge-representatives to enforce pairwise overlapping if needed.   This final brother and edge-representative adjustment is as we have described previously.  

We will consider any vertex-representatives completely contained in that special non-twig path in $t_s$, and show that all but at most one of them can be assigned to twigs in this way:  take any two vertex-representatives $t_u, t_v$ contained in this twig, where w.l.o.g $t_u$ is farther from the single branching node in $t_s$, and $t_v$ is closer to it.  We will call the twigs of the host tree that are attached at branching nodes in $t_{b'}$ the "left" twigs, and those that attach at the single branching node in $t_s$ the "right" twigs.  
Note that any edge-representative for an edge incident at $u$ must have a leaf in $t_u$, similarly with $t_v$.  If $u$ and $v$ both have neighbours in $G$ that are represented on the left twigs as well as the right twigs, then the edge-representatives between $t_u$ and neighbours on the left twigs could not intersect (let along overlap) the edge-representatives between $t_v$ and neighbours on the right twigs, due to the requirement for leaves of these edge-representatives in $t_u$ and $t_v$.  Therefore, as this is a valid representation and so all edge-representatives overlap, it must be that at least one of $t_u, t_v$ has no vertex-representative neighbours in either the left or the right twigs, so we could shift it to the end of a twig, as above. 

No vertex-representatives can be subtrees of $t_{b'}$, as their brothers would then be unrepresentable (they bust overlap the vertex representative, $t_s$ and $t_b$, but not $t_{b'}$).

As in Lemma \ref{lem:moreGeneralList}, there is at most one illegal pair.  Thus, altogether we have one vertex-representative containing $t_b$, one completely contained in a non-twig path within $t_s$, one containing the branching node of $t_s$, and one illegal pair - all the rest of the vertex-representatives are nicely represented on twigs.   Even if each of these inconvenient vertex-representatives belongs to a different copy of $G$, at least two copies are nicely represented, and therefore $G$ is $k$-colourable

\end{proof}

This Lemma ends the proof of Theorem~\ref{t:fixedleaves}. It also directly implies the following:
\begin{theorem}
REC-SUB-T($G$)  is NP-complete.
\end{theorem}
\begin{proof}
If $T$ has three leaves, we are done (proceed as in the previous Theorem). If the graph has only one branching node, we proceed according to Lemma~\ref{l:previous}. Otherwise, for $T$ with $k$ leaves we find its lastbranch of the minimum degree. Let this degree be $d$. For a given graph $G$ (instance of $3$-CON-$k$-COL problem) we create $G_d^{k-d+1}$-blocked graph $G''$ and the rest follows from Lemma~\ref{l:morebranch}.
\end{proof}

Therefore now we have the following corollary:
\begin{corollary}
Given a tree $T$ and a graph $G$, to determine whether $G$ has subtree overlap representation on some subdivision of $T$ is polynomially solvable if and only if $T$ is a path. Given a number $n$ and a graph $G$, to determine whether $G$ has subtree overlap representation on a tree with at most $n$ leaves is polynomial if $n=2$ and NP-complete otherwise.
\end{corollary}
This concludes the part dedicated to subtree overlap representations and we move to subpath representations.

\section{Paths in a Tree}
We now turn to the complexity of REC-PMD-$k(G)$: the problem of recognising overlap graphs of paths in a tree of maximum degree $k$.  

Let $G = (V, E)$ be a connected graph: for fixed integer $k \geq 3$ we will produce a graph $G'' = (V'', E'')$ such that $G''$ is the overlap graph of paths in a tree with maximum degree $k$ if and only if $G$ is $k$-colourable.

Our construction is similar to that for overlapping subtrees, but do not require quite as much machinery in the blocking construct.  

Let the $\emptyset$-blocked graph of a graph $G$ be the $G_d^u$-blocked graph without $V_4$ or any of the edges with an endpoint in $V_4$ - that is, it is just the vertex-representatives, the edge-representatives, and the brothers, without $G_d^u$ itself.

\begin{lemma}
Given a 3-connected graph $G$, its $\emptyset$-blocked graph has a subpath-overlap representation in a tree with maximum degree $k$ if and only if $G$ is $k$-colourable.
It also holds that the $\emptyset$-blocked graph of $G$ has a subpath intersection representation in a tree with maximum degree $k$ if and only if $G$ is $k$-colourable.
\end{lemma}
\begin{proof}
Let $G''$ be the $\emptyset$-blocked graph of $G$.

It is probably not surprising to the reader that we create a canonical representation and show that all representations must be very similar to it. 

Again, we start with showing the representation from the $k$-colouring: first assume that $G$ is $k$-colourable, and let $C_1, ... C_k$ be the colour classes. Now we describe the host tree that we will use to represent $G''$.  We start with a star with $k$ leaves, and replace each of its leaves with an arbitrary subtree of maximum degree three with $n+1$ twigs of length at least 3, where $n$ is the number of vertex-representatives in $G''$ (a long path with pendant twigs is an example of such a tree).  Each of these subtrees will correspond to a colour class in the $k$-colouring.  
We represent each vertex in $C_i$ with a subtree composed of a leaf, its neighbour, and its neighbour's neighbour (that is, a path of length 3 that is a leaf at one end) on the subtree corresponding to $C_i$ - note that these will all be disjoint.  We represent the brother of each vertex-representative by the path from the middle vertex of that vertex-representative's subpath to the central vertex of degree $k$.  We represent each edge-representative with the path between the nodes in its two vertex-representatives that are closest to the central vertex.  This completes our representation (a schematic example version can be found in Figure \ref{fig:p3subpaths}). 

To show that it is valid, note that all vertex-representatives are disjoint, and because the vertex-representatives are all on the ends of twigs, the edge-representatives and brothers are disjoint from all but the vertex-representatives they are concerned with, which they overlap.  Now, because every edge-representative contains a node that is in the subpath of its endpoints, we know that no edge-representative is contained in any other, and they all pairwise intersect, as (because the subtrees correspond to colour classes) every edge-representative is between vertex-representatives on different subtrees, and they all therefore contain the central vertex.  Every brother is defined as containing the central vertex, so all brothers are also pairwise intersecting all other brothers and all edge-representatives.  Because the brothers contain a node closer to the leaf of their appropriate vertex-representative, but do not reach any other subtree, they also are all pairwise overlapping with all other brothers and all edge-representatives.  We can then conclude that this is both an overlap and intersection representation for $G''$.

It remains to show that all representations are almost as nice as this one, and thus that of there is a representation, then $G$ is $k$-colourable.  Here, we try to remain general with respect to it being an overlap or an intersection representation, and when we cannot be, we explain both cases.

Assume that there is an overlap or intersection representation of $G''$ on a host tree with maximum degree $k$.  Due to the Helly-property of subtrees (and therefore subpaths) in a tree, whenever each pair mutually intersects (and to overlap, they have to intersect), there exists a vertex in common between all of these subpaths.  We apply this on the set of edge-representatives and brothers and take a common intersection vertex of highest degree and denote it $c$, which must have degree at most $k$.  
Consider the forest formed by removing $c$ from the host tree: we call the components of this $C_1, C_2, ... C_k$, and will use them to refer to parts of the entire host tree.  Let $u, v$ be two vertices that are adjacent in $G$; we will show that copies of them from the same copy of the vertex set of $G$ cannot be represented on the same one of $C_1, C_2, ... C_k$ in our representation.  We proceed by contradiction: assume that they are represented in the same component; think of the edge-representative of the edge between them in $G$.  It must contain $c$, as well as intersect both of $t_u$ and $t_v$, and it is a path, so we know that one of $t_u$, $t_v$ is "farther" from $c$ - w.l.o.g let it be $t_v$.  Then consider the brother of $t_v$: if it intersects $t_v$ and also contains $c$ (like the edge-representative), then it will also intersect $t_u$ exactly where the edge representative does (and overlap $t_u$, if the edge-representative does), a contradiction.  We can conclude that no two vertex-representative neighbours of a single edge-representative can be represented in the same one of $C_1, C_2, ... C_k$

Now we show that at most two vertex-representatives can contain $c$: first, if we are dealing with an intersection representation, this is clearly impossible, as anything containing $c$ would be adjacent in $G''$ with all edge-representatives and brothers (as well as the other vertex-representatives containing $c$).  If we are dealing with an overlap representation, then consider the brothers of the vertex-representatives.  If there are three vertex-representative subtrees $t_u, t_v, t_w$ that contain $c$, then they must all be in pairwise containment relationships, say $t_u \subset t_v \subset t_w$: then the brother of $t_v$ is unrepresentable, as it will be forced to overlap at least one of $t_u, t_w$ by overlapping $t_v$, or prevented from overlapping the other brothers.  

Then all but two vertex-representatives are represented on one of $C_1, C_2, ... C_k$, and no two that are adjacent (in $G$) are in the same one, therefore at least one copy of $G$ is nicely represented, and $C_1, C_2, ... C_k$ give a $k$-colouring for $G$

\end{proof}

\begin{figure}
\scalebox{0.4}{\includegraphics{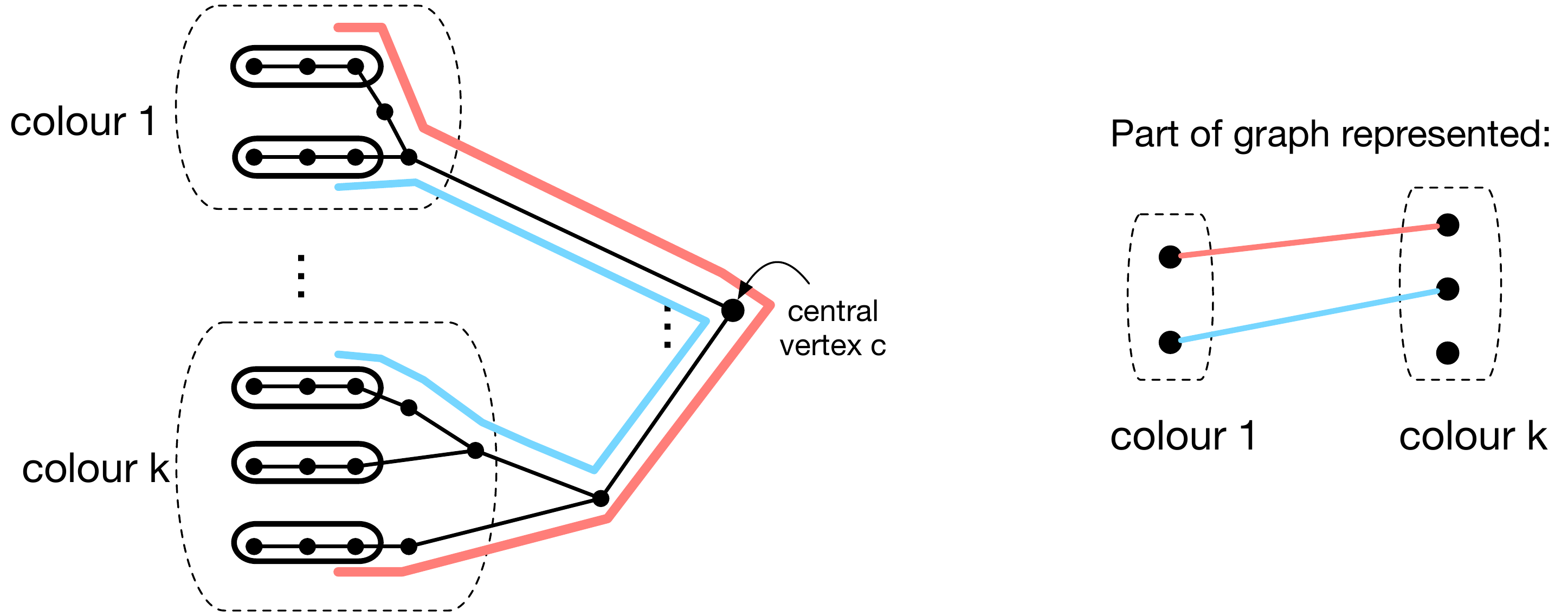}}
\caption{Sample subpath overlap representation on a tree of maximum degree $k \geq 3$. Black ovals are vertex-representatives, the two pale coloured subtrees are edge-representatives. Brothers of vertex-representatives are omitted. Colour classes are shown in dotted ovals.  The part of the graph that is represented is shown on the right.}
\label{fig:p3subpaths}
\end{figure}
So we proved the following two theorems:
\begin{theorem}
REC-PMD-$k(G)$ is NP-complete.
\end{theorem}
\begin{theorem}
For any $k\ge 3$, it is NP-complete to recognise whether a given graph $G$ has a subpath intersection representation in any tree with maximum degree $k$.
\end{theorem}
As for other trees both problems are polynomially solvable (circle- and interval graphs, respectively), we have a dichotomy for these two problems.
\section{Conclusion and Future Work}

We have shown that recognising a number of subclasses of subtree overlap graphs is NP-complete.  This is surprising because of the extreme enforced simplicity of the representations.  Our ultimate goal continues to be resolving the complexity of the recognition problem for subtree overlap graphs in general.  This is currently an open problem.

Other related open problems include tighter bounds on the complicacy of subtree overlap graphs, as well as investigation of other geometric overlap and intersection classes. Also a challenging problem related to subtree overlap graphs (a. k. a. subtree filament graphs) is weighted clique (as Gavril's algorithm a bit surprisingly covers just maximum weighted clique in this case).

Our constructions produce very dense graphs, i.e., graphs with large cliques. A reasonable question is what makes these classes hard to recognise. Is it a presence of large cliques or short cycles like for polygon-circle graphs, or not -- like for, e.g., segment graphs? Is it a large maximum degree, or do these graphs remain hard even for reasonably low degrees, like circle graphs? How about near-planarity?
{%
\bibliographystyle{plain}
\bibliography{sog}}

\begin{thebibliography}{10}

\bibitem{AGG}
Liliana Alcon, Noemi Gudino, and Marisa Gutierrez.
\newblock On k-tree containment graphs of paths in a tree.
\newblock {\em Order}, 2020.

\bibitem{eowynThesis}
Eowyn Cenek.
\newblock Subtree overlap graphs and the maximum independent set problem.
\newblock Master's thesis, University of Alberta, Department of Computing
  Science, 1998.

\bibitem{steveStacho}
Steven Chaplick and Juraj Stacho.
\newblock The vertex leafage of chordal graphs.
\newblock {\em Discrete Applied Mathematics}, 168:14--25, 2014.

\bibitem{corneilOlariuStewart}
Derek~G. Corneil, Stephan Olariu, and Lorna Stewart.
\newblock Lbfs orderings and cocomparability graphs.
\newblock In {\em SODA '99: Proceedings of the tenth annual ACM-SIAM symposium
  on Discrete algorithms}, pages 883--884, Philadelphia, PA, USA, 1999. Society
  for Industrial and Applied Mathematics.

\bibitem{Fel}
Cornelia Dangelmayr and Stefan Felsner.
\newblock Chordal graphs as intersection graphs of pseudosegments.
\newblock In Michael Kaufmann and Dorothea Wagner, editors, {\em Graph
  Drawing}, volume 4372 of {\em Lecture notes in Computer Science}, pages
  208--219. Springer, 2007.

\bibitem{Jess}
Jessica Enright and Lorna Stewart.
\newblock Subtree filament graphs are subtree overlap graphs.
\newblock {\em Inf. Process. Lett.}, 104(6):228--232, 2007.

\bibitem{FMP}
Stefan Felsner, Irina {Musta\,ta}, and Martin Pergel.
\newblock The complexity of the partial order dimension problem: Closing the
  gap.
\newblock {\em SIAM J. Discret. Math.}, 31(1):172--189, 2017.

\bibitem{gavril74}
Fanica Gavril.
\newblock The intersection graphs of subtrees in trees are exactly the chordal
  graphs.
\newblock {\em Journal of Combinatorial Theory (B)}, 16:47--56, 1974.

\bibitem{gavril1978}
Fanica Gavril.
\newblock A recognition algorithm for the intersection graphs of paths in
  trees.
\newblock {\em Disc. Math}, 23:377--388, 1978.

\bibitem{gavril2000}
Fanica Gavril.
\newblock Maximum weight independent sets and cliques in intersection graphs of
  filaments.
\newblock {\em Inf. Process. Lett.}, 73(5-6):181--188, 2000.

\bibitem{GH}
P.C. Gilmore and A.J. Hoffman.
\newblock A characterization of comparability graphs and of interval graphs.
\newblock {\em Canadian Journal of Mathematics}, 16:539 -- 548, 1964.

\bibitem{golumbicBook}
Martin~Charles Golumbic.
\newblock {\em Algorithmic Graph Theory and Perfect Graphs (Annals of Discrete
  Mathematics, Vol 57)}.
\newblock North-Holland, 2004.

\bibitem{golumbicJamison}
Martin~Charles Golumbic and Robert~E. Jamison.
\newblock Edge and vertex intersection of paths in a tree.
\newblock {\em Discrete Mathematics}, 55(2):151--159, 1985.

\bibitem{GL}
Martin~Charles Golumbic and Vincent Limouzy.
\newblock Containment graphs and posets of paths in a tree: wheels and partial
  wheels.
\newblock {\em Order}, 38(1):37--48, 2021.

\bibitem{golumbicLipsteynStern}
Martin~Charles Golumbic, Marina Lipshteyn, and Michal Stern.
\newblock Equivalences and the complete hierarchy of intersection graphs of
  paths in a tree.
\newblock {\em Discrete Applied Mathematics}, 156(17):3203--3215, 2008.

\bibitem{HS}
Michel Habib and Juraj Stacho.
\newblock Linear algorithms for chordal graphs of bounded directed vertex
  leafage.
\newblock {\em Electronic Notes in Discrete Mathematics}, 32:99--108, 2009.

\bibitem{jamisonMulder}
Robert~E. Jamison and Henry~Martyn Mulder.
\newblock Constant tolerance intersection graphs of subtrees of a tree.
\newblock {\em Discrete Mathematics}, 290(1):27--46, 2005.

\bibitem{pergelKratochvil}
Jan Kratochv\'{\i}l and Martin Pergel.
\newblock Two results on intersection graphs of polygons.
\newblock In Giuseppe Liotta, editor, {\em Graph Drawing}, volume 2912 of {\em
  Lecture Notes in Computer Science}, pages 59--70. Springer, 2003.

\bibitem{leclerc}
Bruno Leclerc.
\newblock Arbres et dimension des ordres.
\newblock {\em Discrete Mathematics}, 14(1):69--76, 1976.

\bibitem{extabstr}
J.~Ne\v{s}et\v{r}il M.~\v{S}koviera, editor.
\newblock {\em Recognising the Overlap Graphs of Subtrees of Restricted Trees
  is Hard}, 2019.

\bibitem{McCSpin}
Ross~M. McConnell and Jeremy~P. Spinrad.
\newblock Linear-time transitive orientation.
\newblock In {\em {SODA}}, pages 19--25. {ACM/SIAM}, 1997.

\bibitem{MonmaW86}
Clyde~L. Monma and Victor K.-W. Wei.
\newblock Intersection graphs of paths in a tree.
\newblock {\em J. Comb. Theory, Ser. B}, 41(2):141--181, 1986.

\bibitem{Perm}
Martin Pergel.
\newblock Recognition of polygon-circle graphs and graphs of interval filaments
  is np-complete.
\newblock In Andreas Brandst{\"a}dt, Dieter Kratsch, and Haiko M{\"u}ller,
  editors, {\em WG}, volume 4769 of {\em Lecture Notes in Computer Science},
  pages 238--247. Springer, 2007.

\bibitem{tarjanCan}
Donald~J. Rose, R.~Endre Tarjan, and George~S. Leuker.
\newblock Algorithmic aspects of vertex elimination on graphs.
\newblock {\em SIAM Journal of Computing}, 5(2):266--283, 1976.

\bibitem{rosgenThesis}
William Rosgen.
\newblock Set representations of graphs.
\newblock Master's thesis, University of Alberta, Department of Computing
  Science, 2005.

\bibitem{AAS}
Alejandro~A. Sch\"{a}ffer.
\newblock A faster algorithm to recognize undirected path graphs.
\newblock {\em Discrete Appl. Math.}, 43(3):261--295, 1993.

\bibitem{Spin}
Jeremy~P. Spinrad.
\newblock On comparability and permutation graphs.
\newblock {\em SIAM J. Comput.}, 14:658--670, 1985.

\bibitem{TarjanYannakakis}
Robert~Endre Tarjan and Mihalis Yannakakis.
\newblock Simple linear-time algorithms to test chordality of graphs, test
  acyclicity of hypergraphs, and selectively reduce acyclic hypergraphs.
\newblock {\em SIAM J. Comput.}, 13(3):566--579, 1984.

\bibitem{YAN}
Mihalis Yannakakis.
\newblock The complexity of the partial order dimension problem.
\newblock {\em SIAM J. Alg. Disc. Meth.}, 3:351--358, 09 1982.

\end{thebibliography}

\end{document}